\documentclass[conference,a4paper]{IEEEtran}
\usepackage[ansinew]{inputenc}
\usepackage[pdftex]{graphicx}
\usepackage[T1]{fontenc}
\usepackage{color,amsmath,enumerate,amssymb,hhline,verbatim}
\usepackage{url}
\usepackage{multirow}

\begin{document}

\title{Constructions of Optimal and Almost Optimal Locally Repairable Codes}

\author{
\authorblockN{Toni Ernvall}
\authorblockA{Turku Centre for Computer Science, Turku, Finland\\ \&  Department of Mathematics and Statistics\\ FI-20014 University of Turku, Finland\\ (e-mail: tmernv@utu.fi)}\and
\authorblockN{Thomas Westerb{\"a}ck and Camilla Hollanti }
\authorblockA{Department of Mathematics and Systems Analysis\\ Aalto University, P.O. Box 11100\\ FI-00076 Aalto, Finland \\ (e-mails: \{firstname.lastname@aalto.fi\})}
}

\maketitle

\newtheorem{definition}{Definition}[section]
\newtheorem{thm}{Theorem}[section]
\newtheorem{proposition}[thm]{Proposition}
\newtheorem{lemma}[thm]{Lemma}
\newtheorem{corollary}[thm]{Corollary}
\newtheorem{exam}{Example}[section]
\newtheorem{conj}{Conjecture}
\newtheorem{remark}{Remark}[section]

\newcommand{\La}{\mathbf{L}}
\newcommand{\h}{{\mathbf h}}
\newcommand{\Z}{{\mathbf Z}}
\newcommand{\R}{{\mathbf R}}
\newcommand{\C}{{\mathbf C}}
\newcommand{\D}{{\mathcal D}}
\newcommand{\F}{{\mathbf F}}
\newcommand{\HH}{{\mathbf H}}
\newcommand{\OO}{{\mathcal O}}
\newcommand{\G}{{\mathcal G}}
\newcommand{\A}{{\mathcal A}}
\newcommand{\B}{{\mathcal B}}
\newcommand{\I}{{\mathcal I}}
\newcommand{\E}{{\mathcal E}}
\newcommand{\PP}{{\mathcal P}}
\newcommand{\Q}{{\mathbf Q}}
\newcommand{\M}{{\mathcal M}}
\newcommand{\separ}{\,\vert\,}
\newcommand{\abs}[1]{\vert #1 \vert}

\begin{abstract}
Constructions of optimal locally repairable codes (LRCs) in the case of $(r+1) \nmid n$ and over small finite fields were stated as open problems for LRCs in  [I. Tamo \emph{et al.}, ``Optimal locally repairable codes and connections to matroid theory'', \emph{2013 IEEE ISIT}].  In this paper, these problems are studied by constructing almost optimal linear LRCs, which are proven to be optimal for certain parameters, including cases for which $(r+1) \nmid n$. More precisely, linear codes for given length, dimension, and all-symbol locality are constructed with almost optimal minimum distance. `Almost optimal' refers to the fact that their minimum distance differs by at most one from the optimal value given by a known bound for LRCs. In addition to these linear LRCs, optimal LRCs which do not require a large field are constructed for certain classes of parameters.
\end{abstract}

\section{Introduction}
\subsection{Locally Repairable Codes}
In the literature, three kinds of repair cost metrics are studied: \emph{repair bandwidth} \cite{dimakis}, \emph{disk-I/O} \cite{diskIO}, and \emph{repair locality} \cite{Gopalan,Oggier,Simple}. In this paper the repair locality is the subject of interest.

 Given a finite field $\mathbb{F}_q$ with $q$ elements and an injective function $f: \mathbb{F}_q^k \rightarrow \mathbb{F}_q^n$, let $C$ denote the image of $f$. We say that $C$ is a \emph{locally repairable code (LRC)} and has  \emph{all-symbol locality} with parameters $(n,k,r,d)$, if the code $C$ has minimum (Hamming) distance $d$ and all the $n$ symbols of the code have repair locality $r$. The $j$th symbol has repair locality $s$ if there exists a set
$$
\{i_1,\ldots,i_s\} \subseteq \{1,\ldots,n\} \setminus \{j\}
$$
and a function $f_j$ such that
$$
f_j((y_{i_1}, \ldots, y_{i_s})) = y_j \hbox{ for all } \boldsymbol{y} \in C.
$$
LRCs are defined when $1 \leq r \leq k$. By a linear LRC we mean a linear code of length $n$ and dimension $k$.

In \cite{LRCpapailiopoulos}, Papailiopoulos \emph{et al.} establish an information theoretic bound for both linear and nonlinear codes. With $\epsilon = 0$ in \cite[Thm. 1]{LRCpapailiopoulos} we have the following bound for a locally repairable code $C$ with parameters $(n,k,r,d)$:
\begin{equation}
\label{upperbound}
d \leq n - k - \left \lceil \frac{k}{r} \right \rceil + 2
\end{equation}
A locally repairable code that meets this bound is called \emph{optimal}.




\subsection{Related Work}
As mentioned above, in the all-symbol locality case the information theoretic trade-off between locality and code distance for any (linear or nonlinear) code was derived in \cite{LRCpapailiopoulos}. Furthermore, constructions of optimal LRCs for the case when $(r+1) \nmid n$ and over small finite fields when $k$ is large were stated as open problems for LRCs in \cite{LRCmatroid}. In \cite{LRCmatroid} it was proved that there exists an optimal LRC for parameters $(n,k,r)$ over a field $\mathbb{F}_q$ if $r+1$ divides $n$ and $q=p^{k+1}$ with $p$  large enough. In \cite{SongOptimal} and \cite{TamoBarg} the existence of optimal LRCs was proved for several parameters $(n,k,r)$. Good codes with the weaker assumption of information symbol locality are designed in \cite{Pyramid}. In \cite{Gopalan} it was shown that there exist parameters $(n,k,r)$ for linear LRCs for which the bound of Eq. \eqref{upperbound} is not achievable.

\subsection{Contributions and Organization}
In this paper, we try to build good codes with all-symbol locality, when given parameters $n$, $k$, and $r$. As a measure for the goodness of a code we use its minimum distance $d$. Also, we prefer codes with simple structure, and the property that the construction does not require large field size. Moreover, we give some constructions of optimal LRCs, including cases for which $(r+1) \nmid n$, as well as constructions over small fields. Although codes in the case $(r+1) \nmid n$ are already constructed in \cite{SongOptimal} and \cite{TamoBarg}, the benefits of our construction are that it uses only some elementary linear algebra and it is very simple.

Section \ref{Sec:minDistance} studies the largest achievable minimum distance of the linear locally repairable codes. We show that with a field size large enough we have linear codes with minimum distance at least $d_{\text{opt}}(n,k,r)-1$ for every feasible triplet of parameters $(n,k,r)$. In  Subsection \ref{Subsec:Construction}, we give a construction of such an almost optimal linear locally repairable code. In Subsection \ref{Subsec:Analysis}, we analyze the minimum distance of our construction and derive a lower bound for the largest achievable minimum distance of the linear locally repairable code. Moreover, we prove that our construction results in optimal LRCs (including cases of $(r+1) \nmid n$) for specific parameter values.

In Section \ref{Subsec:Construction-Optimal} we give some constructions of optimal LRCs for certain classes of parameters which do not require a large field. Namely, for certain values of $(r,d)$, we give constructions of optimal $(n,k,r,d)$-LRCs for which the size of the field does not depend on the size of $k$ and $n$.

\section{Constructing Almost Optimal Codes}\label{Sec:minDistance}

\subsection{Construction}\label{Subsec:Construction}
In this subsection we will give a construction for linear locally repairable codes with all-symbol locality over a field $\mathbb{F}_q$ with $q>2\binom{n}{k-1}$, given parameters $(n,k,r)$ such that $n-\left\lceil\frac{n}{r+1}\right\rceil \geq k$. We also assume that $k<n$ and $n \not\equiv 1 \mod r+1$. Write $n=a(r+1)+b$, where $0 \leq b < r+1$.
We will construct a generator matrix for a linear code under the above assumptions. The minimum distance of the constructed code will be studied in Subsection \ref{Subsec:Analysis}.

Next we will build $A=\left\lceil\frac{n}{r+1}\right\rceil$ sets $S_1, S_2, \dots, S_A$ such that each of them consists of $r+1$ vectors of $\mathbb{F}_q^k$, except for $S_A$ that shall consist of $n-(A-1)(r+1)$ vectors of $\mathbb{F}_q^k$.

First, choose any $r$ linearly independent vectors $\mathbf{g}_{1,1},\dots,\mathbf{g}_{1,r}$. Let $\mathbf{s}_{1,r+1}$ be $\sum_{l=1}^{r}\mathbf{g}_{1,l}$. These $r+1$ vectors form the set $S_1$. This set has the property that any $r$ vectors from this set are linearly independent.

Let $1< i \leq A$. Assume that we have $i-1$ sets $S_1, S_2, \dots, S_{i-1}$ such that when taken at most $k$ vectors from these sets, at most $r$ vectors from each set, these vectors are linearly independent. Next we will show inductively that this is possible by constructing the set $S_i$ with the same property.

Let $\mathbf{g}_{i,1}$ be any vector such that when taken at most $k-1$ vectors from the already built sets, with at most $r$ vectors from each set, then $\mathbf{g}_{i,1}$ and these $k-1$ other vectors are linearly independent. This is possible since $\binom{n}{k-1}q^{k-1}<q^k$.
Write $\mathbf{s}_{i,j} = \sum_{l=1}^{j} \mathbf{g}_{i,l}$ for $j=1,\dots,r$.

Suppose we have $j$ vectors $\mathbf{g}_{i,1},\dots,\mathbf{g}_{i,j}$ such that when taken at most $k$ vectors from the sets $S_1, S_2, \dots, S_{i-1}$ or $\{\mathbf{g}_{i,1},\dots,\mathbf{g}_{i,j},\mathbf{s}_{i,j}\}$, with at most $r$ vectors from each set $S_1, S_2, \dots, S_{i-1}$ and at most $j$ vectors from the set $\{\mathbf{g}_{i,1},\dots,\mathbf{g}_{i,j},\mathbf{s}_{i,j}\}$, then these vectors are linearly independent.

Choose $\mathbf{g}_{i,j+1}$ to be any vector with the following two properties: When taken at most $k-1$ vectors from the sets $S_1, S_2, \dots, S_{i-1}$ or $\{\mathbf{g}_{i,1},\dots,\mathbf{g}_{i,j},\mathbf{s}_{i,j}\}$, with at most $r$ vectors from each set $S_1, S_2, \dots, S_{i-1}$ and at most $j$ vectors from the set $\{\mathbf{g}_{i,1},\dots,\mathbf{g}_{i,j},\mathbf{s}_{i,j}\}$, then $\mathbf{g}_{i,j+1}$ and these $k-1$ other vectors are linearly independent. Require also the following property: when taken at most $k-1$ vectors from the sets $S_1, S_2, \dots, S_{i-1}$ or $\{\mathbf{g}_{i,1},\dots,\mathbf{g}_{i,j},\mathbf{s}_{i,j}\}$, with at most $r$ vectors from each set $S_1, S_2, \dots, S_{i-1}$ and at most $j$ vectors from the set $\{\mathbf{g}_{i,1},\dots,\mathbf{g}_{i,j},\mathbf{s}_{i,j}\}$, then $\mathbf{s}_{i,j+1}$ and these $k-1$ other vectors are linearly independent. This is possible because there are at most $\binom{n}{k-1}$ different possibilities to choose, each of the options span a subspace with $q^{k-1}$ vectors, and since $q$ is large we have $2\binom{n}{k-1}q^{k-1}<q^k$. Notice that $\mathbf{s}_{i,j+1} \in V$ (where $V$ is some subspace) if and only if $\mathbf{g}_{i,j+1} \in \mathbf{-s}_{i,j}+V$.

To prove the induction step we have to prove the following thing: when taken at most $k-1$ vectors from sets $S_1, S_2, \dots, S_{i-1}$ or $\{\mathbf{g}_{i,1},\dots,\mathbf{g}_{i,j+1}\}$, with at most $r$ vectors from each set $S_1, S_2, \dots, S_{i-1}$ and at most $j$ vectors from the set $\{\mathbf{g}_{i,1},\dots,\mathbf{g}_{i,j+1}\}$, then $\mathbf{s}_{i,j+1}$ and these $k-1$ other vectors are linearly independent. Let $h \leq j$, $\mathbf{v}$ be a sum of at most $k-1-h$ vectors from the sets $S_1, S_2, \dots, S_{i-1}$ with at most $r$ vectors from each set, and let $m_1<\dots<m_h$ be indices in ascending order. We will assume a contrary: We have coefficients $c_{m_1},\dots,c_{m_h} \in \mathbf{F}_q$ such that
\begin{equation}
\mathbf{s}_{i,j+1} = \mathbf{v} + \sum_{l=1}^{h}c_{m_l}\mathbf{g}_{i,m_l}.
\end{equation}
If $m_h \neq j+1$ then our assumption is false by the definition so assume that $m_h = j+1$.
If $c_{j+1} \neq 1$ then
\begin{equation}
(1-c_{j+1})\mathbf{g}_{i,j+1} = \mathbf{v} + \sum_{l=1}^{h-1}c_{m_l}\mathbf{g}_{i,m_l} - \mathbf{s}_{i,j}.
\end{equation}
and again our assumption is false by the definition.
So assume that $c_{j+1} = 1$. Then we get
\begin{equation}
\mathbf{s}_{i,j} = \mathbf{v} + \sum_{l=1}^{h-1}c_{m_l}\mathbf{g}_{i,m_l}.
\end{equation}
and since $h-1 \leq j-1$ the assumption is false by the induction step.

Now, the sets $S_i$ consist of vectors $\{\mathbf{g}_{i,1},\dots,\mathbf{g}_{i,r},\mathbf{s}_{i,r}\}$ for $i=1,\dots,a$. If $b \neq 0$ the set $S_A$ consists of vectors $\{\mathbf{g}_{A,1},\dots,\mathbf{g}_{A,b-1},\mathbf{s}_{A,b-1}\}$. The matrix $\mathbf{G}$ is a matrix with vectors from the sets $S_1,S_2,\dots,S_A$ as its column vectors, \emph{i.e.},
\[
\mathbf{G}=\left(\mathbf{G}_1|\mathbf{G}_2|\dots|\mathbf{G}_A\right)
\]
where
\[
\mathbf{G_j}=\left(\mathbf{g}_{j,1}|\dots|\mathbf{g}_{j,r}|\mathbf{s}_{j,r}\right)
\]
for $i=1,\dots,a$, and
\[
\mathbf{G_A}=\left(\mathbf{g}_{A,1}|\dots|\mathbf{g}_{A,b-1}|\mathbf{s}_{A,b-1}\right)
\]
if $b\neq 0$.

To be a generator matrix for a code of dimension $k$ the rank of $\mathbf{G}$ has to be $k$. By the construction the rank is $k$ if and only if
$n-A \geq k$, and this is what we assumed.

\subsection{Lower Bound for the Largest Achievable Minimum Distance}\label{Subsec:Analysis}
In this subsection we will derive a lower bound for the largest achievable minimum distance of the linear codes with all-symbol locality. We will do this by analyzing the construction of Subsection \ref{Subsec:Construction}.

 Let $C$ be a linear code with a generator matrix $G$. A subset $A$ of the columns of $G$ is called a \emph{circuit} if $A$ is linearly dependent and all proper subsets of $A$ are linearly independent. A collection of circuits $C_1, \ldots, C_l$ of $C$ is called a \emph{nontrivial union} if
$$
C_i \nsubseteq \bigcup_{j \neq i} C_j, \hbox{ for  } 1 \leq i \leq l.
$$
To analyze our code construction we will use the following result that was proved by Tamo \emph{et al.} in \cite{LRCmatroid}.
\begin{thm}\label{Thm:mindistanceTamo}
The minimum distance of the linear locally repairable code is equal to
\[
d=n-k-\mu+2
\]
where $\mu$ is the minimum positive integer such that the size of every nontrivial union of $\mu$ circuits is at least $\mu+k$.
\end{thm}

To make the notations clearer we define $D_q(n,k,r)$ to be the minimum distance of our code construction for given parameters. To be exact, we have the following definition.
\begin{definition}
If our construction covers parameters $(n,k,r)$ over $\mathbb{F}_q$, then define $D_q(n,k,r)$ to be the minimum distance of such a code. If our construction does not cover parameters $(n,k,r)$ over $\mathbb{F}_q$, then define $D_q(n,k,r)$ to be zero.
\end{definition}

For the largest achievable minimum distance under the assumption of information symbol locality, we mark to be
\[
d_{\text{opt}}(n,k,r) := \max\left\{n-k-\left\lceil\frac{k}{r}\right\rceil+2 , 0 \right\}.
\]
The reason for this kind of definition is that if $n-k-\left\lceil\frac{k}{r}\right\rceil+2 \leq 0$ then it is impossible to have a code for parameters $(n,k,r)$.

Since the assumption of all-symbol locality is stronger than the assumption of information symbol locality, we know that
\begin{equation}\label{upperboundWithq}
D_q(n,k,r) \leq d_{\text{opt}}(n,k,r).
\end{equation}

In \cite{Gopalan} it was proved that there exists triplets $(n,k,r)$ such that the inequality \ref{upperboundWithq} is strict. So the natural question arises:
What is the relationship between $d_{\text{opt}}(n,k,r)$ and $D_q(n,k,r)$? Next we will study this question.


First we need a small straightforward lemma.
\begin{lemma}\label{lemma:kr}
Suppose $n-\left\lceil\frac{n}{r+1}\right\rceil \geq k$. Then $\frac{k}{r} \leq \frac{n}{r+1}$.
\end{lemma}

\begin{proposition}\label{Thm:construction}
Suppose $q>2\binom{n}{k-1}$, $k<n$, $n-\left\lceil\frac{n}{r+1}\right\rceil \geq k$, and $n \not\equiv 1 \mod r+1$. Then
\[
D_q(n,k,r)=d_{\text{opt}}(n,k,r)
\]
if $r+1$ divides $n$, and
\[
D_q(n,k,r) \geq n-k-\left\lfloor\frac{k}{r}-\frac{n}{r+1}\right\rfloor-\left\lfloor\frac{n}{r+1}\right\rfloor
\]
otherwise.
\end{proposition}
\begin{proof}
The construction of Subsection \ref{Subsec:Construction} gives a generating matrix $\mathbf{G}$ for a linear code. It is clear that the code it generates has the all-symbol repair locality $r$.

By Theorem \ref{Thm:mindistanceTamo} its minimum distance is $n-k-\mu+2$ where $\mu$ is a minimum positive integer $m$ with the following property: the size of every nontrivial union of $m$ circuits is at least $k+m$.

We remark that there are circuits of at most two types: possibly of size $k+1$ and those corresponding the sets $S_j$. Suppose we have a nontrivial union of $m$ circuits containing a circuit of size $k+1$. Then the size of this union is at least $(k+1)+(m-1)=k+m$.

Consider now only circuits corresponding the sets $S_j$. We have $A=\left\lceil\frac{n}{r+1}\right\rceil$ such circuits. It is easy to see that every union of such circuits is nontrivial. Write as before $n=a(r+1)+b$ with $0 \leq b < r+1$.

Suppose first that $b=0$. Then $|S_j|=r+1$ for all $j$. Each union of $m$ circuits has the same size
\[
|\cup_{j=1}^{m} S_{i_j}|=m(r+1)
\]
and $m(r+1) \geq m+k$ if and only if $m \geq \frac{k}{r}$, and hence $\mu=\min \left\{\left\lceil\frac{k}{r}\right\rceil ,A+1 \right\}= \left\lceil\frac{k}{r}\right\rceil$ by lemma \ref{lemma:kr}.

This gives that $D_q(n,k,r) =n-k-\left\lceil\frac{k}{r}\right\rceil+2$ when $r+1$ divides $n$.

Suppose now that $b\neq0$. Then $|S_j|=r+1$ for all $j$ except that $|S_A|=b$. Each minimal union of $m$ circuits contains the circuit corresponding the set $S_A$ and hence has the size
\[
|\cup_{j=1}^{m-1} S_{i_j} \cup S_A|=(m-1)(r+1)+b=mr-r+m-1+b.
\]
We have $mr-r+m-1+b \geq m+k$ if and only if
\[
m \geq \frac{k+1+r-b}{r} =1+\frac{k+1-n+\left\lfloor\frac{n}{r+1}\right\rfloor}{r}+\left\lfloor\frac{n}{r+1}\right\rfloor.
\]
Notice also that
\begin{equation}
\begin{split}
& \left\lceil 1+\frac{k+1-n+\left\lfloor\frac{n}{r+1}\right\rfloor}{r}+\left\lfloor\frac{n}{r+1}\right\rfloor \right\rceil \\
= & \left\lfloor\frac{k}{r}-\frac{n}{r+1}\right\rfloor+\left\lfloor\frac{n}{r+1}\right\rfloor+2 \\
\end{split}
\end{equation}
and hence
\begin{equation}
\begin{split}
\mu & =\min\left\{ \left\lfloor\frac{k}{r}-\frac{n}{r+1}\right\rfloor+\left\lfloor\frac{n}{r+1}\right\rfloor+2 , A+1 \right\} \\
& =\left\lfloor\frac{n}{r+1}\right\rfloor+2+\left\lfloor\frac{k}{r}-\frac{n}{r+1}\right\rfloor \\
\end{split}
\end{equation}
by lemma \ref{lemma:kr}.

This gives that
\[
D_q(n,k,r) \geq n-k-\mu+2 = n-k-\left\lfloor\frac{k}{r}-\frac{n}{r+1}\right\rfloor-\left\lfloor\frac{n}{r+1}\right\rfloor
\]
when $n \not\equiv 0,1 \mod r+1$.

\end{proof}

As a consequence the above analysis of the construction we have the following theorem.
\begin{thm}\label{Thm:MinDistanceSet}
Suppose $q>2\binom{n}{k-1}$ and $k<n$. Then $D_q(n,k,r) \in \{d_{\text{opt}}(n,k,r)-1,d_{\text{opt}}(n,k,r)\}$.
\end{thm}
\begin{proof}
Write $n=a(r+1)+b$ with $0 \leq b < r+1$.

Suppose first that
\begin{equation}\label{Eq:NotConstruction}
n-\left\lceil\frac{n}{r+1}\right\rceil +1 \leq k.
\end{equation}

If $r+1$ divides $n$ then the Equation \ref{Eq:NotConstruction} has the form
$
ar +1 \leq k
$
and hence
\[
n-k-\left\lceil \frac{k}{r}\right\rceil+2 \leq a(r+1)-(ar+1)-(a+1)+2=0.
\]
So it is impossible to have a code for parameters $(n,k,r)$ and hence $D_q(n,k,r)=d_{\text{opt}}(n,k,r)=0$.

If $r+1$ does not divide $n$ then the Equation \ref{Eq:NotConstruction} has the form
$
ar+b \leq k
$
and hence
\[
n-k-\left\lceil \frac{k}{r}\right\rceil+2 \leq a(r+1)+b-(ar+b)-(a+1)+2  \leq 1.
\]
Hence $D_q(n,k,r) \geq 0 \geq d_{\text{opt}}(n,k,r)-1$.

Suppose then that $n-\left\lceil\frac{n}{r+1}\right\rceil \geq k$. Now we can use Theorem \ref{Thm:construction}.

If $b=0$ then the claim is true by the Proposition \ref{Thm:construction}.

Assume $b=1$ and $\mathbf{G}$ is a generating matrix of a linear locally repairable code for parameters $(n-1,k,r)$ and minimum distance $d_{\text{opt}}(n-1,k,r)$. Replicate any column in $\mathbf{G}$ and get a generating matrix for a linear locally repairable code for parameters $(n,k,r)$ and minimum distance $d_{\text{opt}}(n,k,r)-1$.

Assume $b>1$. Then
\[
D_q(n,k,r) \geq n-k-\left\lfloor\frac{k}{r}-\frac{n}{r+1}\right\rfloor-\left\lfloor\frac{n}{r+1}\right\rfloor
\]
and hence
\begin{equation}\label{Eq:optimijaalaraja}
\begin{split}
& d_{\text{opt}}(n,k,r) - D_q(n,k,r) \\
\leq & \left\lfloor\frac{k}{r}-\frac{n}{r+1}\right\rfloor -\left\lceil\frac{k}{r}\right\rceil + \left\lfloor\frac{n}{r+1}\right\rfloor +2 \\
\leq & \left\lfloor-\frac{n}{r+1}\right\rfloor + \left\lfloor\frac{n}{r+1}\right\rfloor +2 =  -1 +2 =1.\\
\end{split}
\end{equation}

\end{proof}

So we know that if $d_{\text{opt}}(n,k,r) \geq 2$ then we have a linear locally repairable code for parameters $(n,k,r)$. If $d_{\text{opt}}(n,k,r) = 0$ then it is impossible to have a linear locally repairable code for parameters $(n,k,r)$. However, if $d_{\text{opt}}(n,k,r) = 1$ then we do not know whether there exists a linear locally repairable code for parameters $(n,k,r)$.

Theorem \ref{Thm:MinDistanceSet} gives a lower bound for the minimum distance. In fact we can say little more in a certain case.

Below, the \emph{fractional part} of of $x$ is denoted by $\{x\}$, \emph{i.e.},  $\{x\} = x - \lfloor x\rfloor$.
\begin{thm}\label{Thm:MinDistanceWithFracPart}
Suppose $q>2\binom{n}{k-1}$, $k<n$, $\left\{\frac{k}{r}\right\} < \left\{\frac{n}{r+1}\right\}$, and $r$ does not divide $k$. Then
\[
D_q(n,k,r)=d_{\text{opt}}(n,k,r).
\]
\end{thm}
\begin{proof}
Write $n=a(r+1)+b$ with $0 \leq b < r+1$. If $b= 0 \text{ or } 1$ then $\left\{\frac{k}{r}\right\} \geq \left\{\frac{n}{r+1}\right\}$ so we may assume that this is not the case.

Suppose first that $n-\left\lceil\frac{n}{r+1}\right\rceil \geq k$. By studying the Equation \ref{Eq:optimijaalaraja} again we notice that
\begin{equation}\label{Eq:optimijaalarajaTarkemmin}
\begin{split}
& d_{\text{opt}}(n,k,r) - D_q(n,k,r) \\
\leq & \left\lfloor\frac{k}{r}-\frac{n}{r+1}\right\rfloor -\left\lceil\frac{k}{r}\right\rceil + \left\lfloor\frac{n}{r+1}\right\rfloor +2 \\
= & \left\lfloor\left\{\frac{k}{r}\right\}-\left\{\frac{n}{r+1}\right\}\right\rfloor +1 =0 \\
\end{split}
\end{equation}

Suppose then that $n-\left\lceil\frac{n}{r+1}\right\rceil +1 \leq k$. Since $\left\{\frac{k}{r}\right\} < \left\{\frac{n}{r+1}\right\}$ we know that $r+1$ cannot divide $n$ and hence we have $ar+b \leq k$. It is impossible that $ar+b = k$. Indeed, then we would have $\left\{\frac{k}{r}\right\}=\left\{\frac{b}{r}\right\} > \left\{\frac{b}{r+1}\right\}=\left\{\frac{n}{r+1}\right\}$. Hence $ar+b < k$ and
\begin{equation}
\begin{split}
d_{\text{opt}}(n,k,r) & = \max\left\{n-k-\left\lceil\frac{k}{r}\right\rceil+2,0\right\} \\
& \leq \max\left\{1-\left\lceil\frac{b+1}{r}\right\rceil,0\right\}=0
\end{split}
\end{equation}
and hence $D_q(n,k,r)=d_{\text{opt}}(n,k,r)$.
\end{proof}

\section{Constructing Optimal LRCs over $\mathbb{F}_4$} \label{Subsec:Construction-Optimal}

In this section we give some constructions of optimal LRCs over the field of four elements $\mathbb{F}_4$ for certain values of $(r,d)$. Our LRCs will be described in the setting of matrices with different operators as entries.

\subsection{Matrix Representation}

We represent the elements of $\mathbb{F}_4$ as $\{00,01,10,11\}$, such that the addition of elements in $\mathbb{F}_4$ can be considered as bitwise addition without carry (e.g. $01 + 11 = 10$). In our construction of optimal LRCs over $\mathbb{F}_4$ we will use the operators $\alpha$, $\alpha^2$, $\beta$, $\beta^2$, $\beta^3$, $1 \hbox{ and } 0$ on $\mathbb{F}_4$ to $\mathbb{F}_4$ defined as
$$
\begin{array}{rrrr}
\alpha(00) = 00, & \alpha(01) = 10, & \alpha(10) = 11 & \alpha(11) = 01,\\
\alpha^2(00) = 00, & \alpha^2(01) = 11, & \alpha^2(10) = 01 & \alpha^2(11) = 10,\\
\beta(00) = 01, & \beta(01) = 10, & \beta(10) = 11 & \beta(11) = 00,\\
\beta^2(00) = 10, & \beta^2(01) = 11, & \beta^2(10) = 00 & \beta^2(11) = 01,\\
\beta^3(00) = 11, & \beta^3(01) = 00, & \beta^3(10) = 01 & \beta^3(11) = 10,\\
1(00) = 00, & 1(01) = 01, & 1(10) = 10 & 1(11) = 11,\\
0(00) = 00, & 0(01) = 00, & 0(10) = 00 & 0(11) = 00.
\end{array}
$$

A code $C$ is represented by a $k \times n$ matrix $F$. The entries of the matrix are the operators $\alpha, \alpha^2, \beta, \beta^2, \beta^3, 1 \hbox{ and } 0$. The code $C$ consists of the following codewords
$$
C = \{\boldsymbol{y} \in \mathbb{F}_4^n : y_j = F_{1,j}(x_1) + \ldots + F_{k,j}(x_k) \hbox{ for } \boldsymbol{x} \in \mathbb{F}_4^k\}.
$$

\subsection{Optimal LRCs over $\mathbb{F}_4$}

Let $A$ and $B$ be the following matrices
$$
A = \left (
\begin{array}{cccc}
1 & 0 & 0 & 1 \\
0 & 1 & 0 & 1 \\
0 & 0 & 1 & 1
\end{array}
\right )
\quad \hbox{and} \quad
B = \left (
\begin{array}{ccc}
0 & 1 & 1 \\
0 & \alpha & \alpha \\
0 & \alpha^2 & \alpha^2
\end{array}
\right ).
$$
For $i \geq 1$, let $F_i^1(3,3)$ be the $(i+1) \times (i+1)$-block matrix
$$
\begin{small}
F_i^1(3,3) =
\left (
\begin{array}{c | c | c | ccc}
 A & & & & B   & \\
\hline
& \ddots & & & \vdots & \\
\hline
& & A  &  & B & \\
\hline
&&& 1 & \alpha & \alpha^2
\end{array}
\right ),
\end{small}
$$
where the entires of the empty blocks are 0-operators and the first $i$ diagonal blocks are $A$-blocks.

\begin{thm} \label{th:linear-optimal-(3,3)}
The matrix $F_i^1(3,3)$ defines a locally repairable code $C$ over $\mathbb{F}_4$ with parameters $(n,k,d,r) = (4i+3, 3i+1,3,3)$ for $i \geq 1$.
\end{thm}
\begin{proof}
Let $f:\mathbb{F}_4^{3i+1} \rightarrow \mathbb{F}_4^{4i+3}$ denote the mapping given by the matrix $F_i^1(3,3)$. Now, $f$ is injective, because
\begin{equation} \label{eq:identity_map_in_f}
\begin{array}{l}
y_{4j-3} = x_{3j-2}, y_{4j-2} = x_{3j-1}, y_{4j-1} = x_{3j} \hbox{ and}\\
y_{4i+1} = x_{3i+1},
\end{array}
\end{equation}
for $1 \leq j \leq i$ and $f(\boldsymbol{x}) = \boldsymbol{y}$. Since $f$ is injective and $F_i^1(3,3)$ is a $(3i+1) \times (4i+3)$-matrix it follows that $(n,k) = (4i+3,3i+1)$.

The code $C$ has repair locality $r=3$, since from the fact that $1(x) + \alpha(x) + \alpha^2(x) = 00$ for $x \in \mathbb{F}_4$ we may deduce that
$$
y_{4i+1} + y_{4i+2} + y_{4i+3} = 00,
$$
for $f(\boldsymbol{x}) = \boldsymbol{y}$. Moreover, we have that
$$
y_{4j-3} + y_{4j-2} + y_{4j-1} + y_{4j} = 00,
$$
for $f(\boldsymbol{x}) = \boldsymbol{y}$ and $1 \leq j \leq i$.

Let $d(\boldsymbol{u}, \boldsymbol{v})$ denote the distance between $\boldsymbol{u}, \boldsymbol{v} \in \mathbb{F}_4^m$.
Suppose $\boldsymbol{w}$ and $\boldsymbol{x}$ are two elements of $\mathbb{F}_4^{3i+1}$ such that $d(\boldsymbol{w}, \boldsymbol{x}) \geq 3$. Then, by (\ref{eq:identity_map_in_f}), we deduce that $d(f(\boldsymbol{w}), f(\boldsymbol{x})) \geq 3$.

Suppose $\boldsymbol{w}$ and $\boldsymbol{x}$ are two elements of $\mathbb{F}_4^{3i+1}$ such that $d(\boldsymbol{w}, \boldsymbol{x}) = 1$. We note that every row of $F_i^1(3,3)$ has at least three entries $a$, $b$ and $c$ with operators  1, $\alpha$ or $\alpha^2$. In particular, this yields that the coefficients $a$, $b$ and $c$ of $f(\boldsymbol{w})$ differ from these coefficients of $f(\boldsymbol{x})$, and hence $d(f(\boldsymbol{w}), f(\boldsymbol{x})) \geq 3$.

Suppose $\boldsymbol{w}$ and $\boldsymbol{x}$ are two elements of $\mathbb{F}_4^{3i+1}$ such that $d(\boldsymbol{w}, \boldsymbol{x}) = 2$. Let $a$ and $b$ be the index of the two coefficients in which $\boldsymbol{w}$ and $\boldsymbol{x}$ differ. Assume that row $a$ and row $b$ of $F_i^1(3,3)$ are in different horizontal blocks. Then there are at least three columns $e$, $g$ and $h$ of $F_i^1(3,3)$ such that one of the entries $(a,e)$ and $(b,e)$ is the 0-operator and the other one is the 1-operator, this property also holds for the entries $(a,g)$, $(b,g)$ and $(a,h)$, $(b,h)$. Consequently, $d(f(\boldsymbol{w}), f(\boldsymbol{x})) \geq 3$ when the rows $a$ and $b$ are in different horizontal blocks.

Now, suppose that row $a$ and row $b$ are in the same horizontal block of $F_i^1(3,3)$, i.e. row $a$ and $b$ are rows in a submatrix of the following form
$$
\left (
\begin{array}{c|c|c|c}
 \mathbf{0}& A & \mathbf{0} & B
 \end{array}
 \right).
$$
It is easy to check by hand that if
$$
y \neq y' \hbox{, } z \neq z' \hbox{ and } 1(y) + 1(z) = 1(y') + 1(z'),
$$
then
$$
\begin{array}{l}
1(y) + \alpha(z) \neq 1(y') + \alpha(z'),\\
1(y) + \alpha^2(z) \neq 1(y') + \alpha^2(z'),\\
\alpha(y) + \alpha^2(z) \neq \alpha(y') + \alpha^2(z').
\end{array}
$$
for $y,y',z,z' \in \mathbb{F}_4$. As a consequence of this fact and since there are two pair of entries $\{(a,g),(b,g)\}$ and $\{(a,h),(b,h)\}$ in $F_i^1(3,3)$ such that one of the entries in each pair is the 0-operator and the other entry is the 1-operator, we deduce that $d(f(\boldsymbol{w}), f(\boldsymbol{x})) \geq 3$. Hence $d \geq 3$ for the code.

Moreover, since
$$
4i + 3 - (3i + 1) - \left \lceil \frac{3i+1}{3} \right \rceil + 2 = 3
$$
we obtain that $C$ is an optimal $(4i+3,3i+1,3,3)$-LRC.
\end{proof}

Let $D$ be the following matrix
$$
D = \left (
\begin{array}{cccc}
0 & 0 & 1 & 1 \\
0 & 0 & \alpha & \alpha \\
0 & 0 & \alpha^2 & \alpha^2
\end{array}
\right ).
$$
For $i \geq 1$, let $F_i^2(3,3)$ be the $(i+1) \times (i+1)$-block matrix
$$
F_i^2(3,3) =
\left (
\begin{array}{c | c | c |cccc}
 A & & & & D   & &\\
\hline
& \ddots & & & \vdots & & \\
\hline
& & A  &  & D & & \\
\hline
&&& 1 & 0 & \beta & \beta^2\\
&&& 0 & 1 & \beta^2 & \beta
\end{array}
\right ),
$$
and let $F_i^1(3,4)$ be the $(i+1) \times (i+1)$-block matrix
$$
F_i^1(3,4) =
\left (
\begin{array}{c | c | c |cccc}
 A & & & & A   & &\\
\hline
& \ddots & & & \vdots & & \\
\hline
& & A  &  & A & & \\
\hline
&&& 1 & \beta & \beta^2 & \beta^3
\end{array}
\right ).
$$

With similar proof techniques as in Theorem \ref{th:linear-optimal-(3,3)} we can prove the following two theorems.

\begin{thm} \label{th:nonlinear-optimal-(3,3)}
The matrix $F_i^2(3,3)$ defines a locally repairable code $C$ over $\mathbb{F}_4$ with parameters $(n,k,d,r) = (4i+4, 3i+2,3,3)$ for $i \geq 1$.
\end{thm}

\begin{thm} \label{th:nonlinear-optimal-(3,4)}
The matrix $F_i^1(3,4)$ defines a locally repairable code $C$ over $\mathbb{F}_4$ with parameters $(n,k,d,r) = (4i+4, 3i+1,3,4)$ for $i \geq 1$.
\end{thm}

Note that the codes we construct in Theorem \ref{th:nonlinear-optimal-(3,3)} and Theorem \ref{th:nonlinear-optimal-(3,4)} are nonlinear since $\beta$ is a nonlinear operator over $\mathbb{F}_4$.

\section{Future Work}


As future work it is still left to find the exact expression of the largest achievable minimum distance of a linear locally repairable code with all-symbol locality when given the length $n$, dimension $k$, and locality $r$ of the code. Our goal is to also generalize the constructions given in Section \ref{Subsec:Construction-Optimal} to other parameters $r$ and $d$ over small fields.


\end{document}